\documentclass{article}
\usepackage{graphicx,color} 

\usepackage{geometry}
\usepackage{setspace,amsmath,amssymb,amsthm,url,hyperref}
\usepackage{cleveref}
\usepackage{natbib}
\usepackage{authblk}
\usepackage[table,dvipsnames]{xcolor}
\usepackage{tcolorbox}
\newtcolorbox[auto counter]{mybox}[2][]{float,title={\textcolor{black}{Algorithm~\thetcbcounter: #2}},#1, colframe=white!70!gray}

\newtheorem{corollary}{Corollary}
\newtheorem{theorem}{Theorem}
\newtheorem{remark}{Remark}

\def\trace {{\rm trace}}

\newcommand{\bmka}{0}

\title{Thinning a Wishart Random Matrix}
\author[1]{Ameer Dharamshi}
\author[2]{Anna Neufeld}
\author[3]{Lucy L. Gao}
\author[1,4]{Daniela Witten}
\author[5]{Jacob Bien}
\affil[1]{Department of Biostatistics, University of Washington}
\affil[2]{Department of Mathematics and Statistics, Williams College}
\affil[3]{Department of Statistics, University of British Columbia}
\affil[4]{Department of Statistics, University of Washington}
\affil[5]{Department of Data Sciences and Operations, University of Southern California}

\begin{document}

\maketitle

\begin{abstract}
    Recent work has explored \emph{data thinning}, a generalization of sample splitting that involves decomposing a (possibly matrix-valued) random variable into independent components. In the special case of a $n \times p$ random matrix with independent and identically distributed $N_p(\mu, \Sigma)$ rows, \citet{dharamshi2024decomposing} provides a comprehensive analysis of the settings in which thinning is or is not possible: briefly, if $\Sigma$ is unknown, then one can thin provided that $n>1$.   However, in some situations a data analyst may not have direct access to the data itself.  For example, to preserve individuals' privacy, a data bank may  provide only summary statistics such as the sample mean and sample covariance matrix.  
     While the sample mean follows a Gaussian distribution,  the sample covariance follows (up to scaling) a  Wishart distribution, for which no thinning strategies have yet  been proposed.  In this note, we fill this gap:  we show that it is possible to generate two independent data matrices with independent $N_p(\mu, \Sigma)$ rows, based only on the sample mean and sample covariance matrix. These independent data matrices can either be used directly within a train-test paradigm, or can be used to derive independent summary statistics. Furthermore, they can be recombined to yield the original sample mean and sample covariance.
\end{abstract}

\section{Introduction}

Many modern data analysis pipelines rely on the ability to split a dataset into independent parts. For instance, one might wish to fit a  model to one part and validate it on another, or else to select a parameter of interest on one part and conduct inference on the other.  
In cases where we have access to $n$ independent and identically distributed observations, \emph{sample splitting}  provides a simple strategy to split our data into $K \leq n $ independent parts \citep{cox1975note}.
 However, in some settings, sample splitting is either inapplicable or unattractive. For instance, perhaps the $n$ observations are not independent or not identically distributed, or perhaps $n=1$. 

In a recent line of work, a number of authors have considered alternatives to sample splitting that involve splitting a \emph{single} (possibly matrix-valued) random variable into independent random variables, which can be recombined to yield the original random variable. We refer to such strategies, in aggregate, as \emph{data thinning}: see Definition 1 of \cite{dharamshi2023generalized}. 
\cite{robins2006adaptive}, \cite{tian2018selective}, \cite{rasines2021splitting}, and \cite{leiner2022data} show that it is possible to thin a $N_p(\mu,\Sigma)$ random vector with $\mu$ unknown and $\Sigma$ known. 
 \cite{neufeld2023data} extended this strategy to natural exponential families, such as the binomial family and the negative binomial family with known overdispersion parameter. \cite{dharamshi2023generalized} clarified the class of distributions that can be thinned, and showed that it extends far beyond natural exponential families, to examples such as the uniform and beta families.

 \cite{dharamshi2024decomposing} outlines the following possibilities
for thinning 
$n$ independent and identically distributed  $N_p(\mu, \Sigma)$ random variables:
\begin{list}{}{}
\item{\emph{Case 1: $\Sigma$ is known, $n \geq 1$.}} \cite{rasines2021splitting} and \cite{tian2018selective} provide a thinning strategy.
\item{\emph{Case 2: $\Sigma$ is unknown, $n > 1$. }} Proposition 4 of \cite{dharamshi2024decomposing} provides a thinning strategy.
\item{\emph{Case 3: $\Sigma$ is unknown, $n=1$.} } \cite{dharamshi2024decomposing} prove that if $p>1$,  thinning is not possible. 
\end{list}

In contrast to prior work, in this note we consider a situation in which we do not actually have access to the original sample of $N_p(\mu, \Sigma)$  random variables: that is, $Z_1,\ldots,Z_n \sim N_p(\mu, \Sigma)$ are unobserved. Instead, we only have access to the summary statistics, the sample mean $\bar Z_n$ and the sample covariance $S_n$: 
\begin{align}\label{eq:suff-unknown-mean}
    \bar Z_n = \frac{1}{n}\sum_{i=1}^n Z_i, \qquad S_n=\frac{1}{n-1}\sum_{i=1}^n(Z_i-\bar Z_n)(Z_i-\bar Z_n)^\top.
\end{align}

This may be the case for one of the following reasons:
\begin{enumerate}
    \item \emph{Privacy considerations may preclude the release of $Z_1,\ldots,Z_n$; however, $\bar Z_n$ and $S_n$ can be released}. For instance, in the context of genetic data, it is often not possible to share the raw data. Instead, summary statistics --- which are typically less personally identifiable --- of the data are shared. As one example, \cite{pasaniuc2017dissecting} discuss the release of a correlation matrix between genetic variants in cases where individual-level data cannot be shared due to privacy concerns. 
    \item \emph{The data analysis pipeline requires only the summary statistics,} and the data analyst does not have access to the original data $Z_1,\ldots, Z_n$. This may be due to scientific considerations: for example, in the context of neuroscience research, analyses often center on the $p \times p$ matrix of connectivity between voxels of the brain \citep{cohen2017computational}.  Or it might be due to statistical considerations: for example, the graphical lasso proposal \citep{friedman2008sparse} operates on the sample covariance matrix, not the matrix normal data matrix from which it arose. Or alternatively, perhaps the $p \times p$ matrix $S_n$ was measured directly, i.e. there is no $Z_1,\ldots,Z_n$, as in classical multidimensional scaling \citep{torgerson1952multidimensional}. 
    
\end{enumerate}

Classical results in multivariate statistics tell us that $\bar Z_n \sim N_p (\mu, \Sigma/n)$ and  $(n-1) S_n \sim \text{Wishart}_p(n -1, \Sigma)$, the $(p \times p)$-dimensional Wishart distribution with $n-1$ degrees of freedom (see Remark \ref{rem:swishart}). 
In this note, we develop a thinning strategy to create two or more independent random matrices with independent $N_p(\mu,\Sigma)$ rows from these summary statistics.  
  The key technical result making this possible is a procedure, which we introduce in Section~\ref{sec:proposal}, that is  originally due to \cite{lindqvist2005monte}. We go on to show how this result can be used to thin a Wishart distribution into two (or more) Wisharts, thereby adding a new entry into the list of natural exponential families where convolution-closed thinning \citep{neufeld2023data} is known to be possible. 

Henceforth, we will use the notation $N_{a \times b} (M, \Delta, \Gamma)$ to denote the matrix normal distribution with $a$ rows, $b$ columns, $a \times b$ mean matrix $M$, $a \times a$ row-covariance matrix $\Delta$, and $b \times b$ column-covariance matrix $\Gamma$. 
Moreover, we will use the notation $\text{Unif}(O_{k \times l})$ to indicate the uniform distribution on the set of orthogonal $k \times l$ matrices. This is  known as the Haar invariant distribution (on $O_{k\times l}$)  \citep{anderson2003introduction, muirhead2009aspects}. 

\begin{remark}
\label{rem:swishart}
When $n \le p$, $(n-1)S_n$ follows a \emph{singular} Wishart distribution \citep{srivastava2003singular}; the distinction between the singular and non-singular Wishart distributions is not important in what follows and thus we will use the word ``Wishart" throughout.
\end{remark}

\section{A matrix square root of a Wishart with independent Gaussian rows}\label{sec:proposal}

Given a rank-$r$ matrix $W\in\mathbb R^{p\times p}$, if the $n \times p$ matrix $A$ satisfies $A^\top A=W$, then we say that $A$ is a matrix square root of $W$. (Of course, it must be the case that   $n \geq r$.)
 The matrix square root is not unique. 
 For example, consider the eigenvalue decomposition $W=VD^2V^\top$, where $D$ is a $r \times r$ diagonal matrix and $V$ is a $p \times r$ orthogonal matrix: then  for any $r \times r$ orthogonal matrix $Q$, it follows that  $QDV^\top$ is a matrix square root of $W$. 

By definition, a Wishart random matrix $W$ has a matrix square root with rows that are independent and identically distributed multivariate Gaussians.  One might hope that any matrix square root of a Wishart matrix would have this property, but this is not the case (see, e.g. Section~\ref{sec:verification}).
To achieve this property, we present Algorithm~\ref{alg1}. Theorem~\ref{thm:root} that follows shows that this algorithm generates matrix square roots with independent and identically distributed Gaussian rows. 

\if\bmka1
\begin{algo}
\label{alg1} \textcolor{white}{.}\\
\indent \emph{Input:} a $p \times p$ positive semi-definite matrix $W$ of rank $r$ and an integer $n \geq r$. 
\begin{enumerate}
    \item Perform an eigenvalue decomposition: $W=VD^2V^\top$, where $D$ is a $r \times r$ diagonal matrix, and $V$ is an orthogonal matrix of dimension $p\times r$.
    \item Draw $Q\sim\text{Unif}(O_{n\times r})$, where $O_{n\times r}=\{Q\in\mathbb{R}^{n\times r}: Q^\top Q=I_r\}$. 
    \item Return the $n\times p$ matrix $X=QDV^\top$, with rows $X_1,\ldots,X_n \in \mathbb{R}^p$.
\end{enumerate}
\end{algo}
\else
\begin{mybox}[floatplacement=!h,label={alg1}]{ Decomposing a $p \times p$ positive semi-definite matrix $W$ of rank $r$ into an $n \times p$ matrix  $X$, for some $n \geq r$} 
\begin{enumerate}
    \item Perform an eigenvalue decomposition: $W=VD^2V^\top$, where $D$ is a $r \times r$ diagonal matrix, and $V$ is an orthogonal matrix of dimension $p\times r$.
    \item Draw $Q\sim\text{Unif}(O_{n\times r})$, where $O_{n\times r}=\{Q\in\mathbb{R}^{n\times r}: Q^\top Q=I_r\}$. 
    \item Return $X=QDV^\top$, with rows $X_1,\ldots,X_n \in \mathbb{R}^p$.
\end{enumerate}
\end{mybox}
\fi

\begin{theorem}[A square root of a  Wishart  with independent Gaussian rows]
\label{thm:root}

Suppose that we apply  Algorithm~\ref{alg1} to $(W, n)$, where $W\sim \text{Wishart}_{p}(n,\Sigma)$, to obtain an $n \times p$ matrix $X$. 
Then, $X^\top X = W$, and  the rows of $X$ are independent $N_p(0, \Sigma)$ random variables.
\end{theorem}

The proof of Theorem \ref{thm:root} is given in Supplement \ref{sec:appendixA}.

In the next section, we will show that  Theorem~\ref{thm:root}  can be applied to thin the summary statistics of an unobserved sample of independent and identically distributed Gaussian vectors.

\section{Thinning the sample covariance}

We return now to the setting of this paper, where $Z_1,\ldots,Z_n\sim N_p(\mu,\Sigma)$ denote a sample of $n$ independent Gaussian vectors that are {\em unavailable} to the data analyst.

\subsection{The case where $\mu$ is known}

We first consider the case where $\mu$ is known, and the  analyst is provided with

\begin{align}\label{eq:suff-known-mean}
    \tilde S_n=\frac{1}{n}\sum_{i=1}^n(Z_i-\mu)(Z_i-\mu)^\top
\end{align}
along with the sample size $n$.  
The following corollary of Theorem~\ref{thm:root} enables us to thin $\tilde S_n$ into independent Wishart random matrices. 
\begin{corollary}[Thinning the sample covariance  of independent Gaussians with known mean]
\label{cor:knownmean}
Suppose that we apply Algorithm~\ref{alg1}  to $(n \tilde S_n, n)$  to obtain an $n \times p$ matrix $X$, where $\tilde S_n$  is defined in \eqref{eq:suff-known-mean} for $Z \sim N_{n \times p} (1_n \mu^\top; I_n, \Sigma)$. Then, (i) $X^\top X=n\tilde S_n$, and (ii) the rows of $X$ are independent $N_p(0, \Sigma)$ random variables. Furthermore, let $C_1,\ldots,C_K$ denote a partition of the integers $\{1,\ldots,n\}$ such that $C_k \cap C_{k'} = \emptyset$ for any $k \neq k'$ and $\cup_{k=1}^K C_k = \{1,\ldots,n\}$, and define
 $S^{(k)} := \frac{1}{|C_k|} \sum_{i \in C_k} X_i X_i^\top $ where $|C_k|$ is the number of elements in the set $C_k$. Then, (iii) $|C_k| S^{(k)} \sim \text{Wishart}_p(|C_k|, \Sigma)$ and $S^{(1)},\ldots,S^{(K)}$ are independent.
\end{corollary}
\begin{proof}
Noting that $n \tilde S_n \sim \text{Wishart}_p(n, \Sigma)$, (i) and (ii) follow immediately from Theorem~\ref{thm:root}. Furthermore, (iii) follows from the independence of the rows of $X$, as well as the definition of the Wishart distribution. 
\end{proof}

What is the point of Corollary~\ref{cor:knownmean}? Given the sample covariance matrix from a sample of $n$ independent $N_p(\mu, \Sigma)$ random vectors, we can obtain either (a) $K$ independent normal data matrices $X^{(k)}\sim N_{n_k\times p}(1_{n_k}\mu^T; I_{n_k},\Sigma)$, where $n_1+\dots+n_k=n$, or (b) $K$ independent sample covariance matrices corresponding to those data matrices.
We can use (a) in order to conduct a data analysis pipeline, such as cross-validation, that requires multiple independent data folds. We can use (b) if the data analysis pipeline specifically requires sample covariance matrices. In either case, the $K$ independent random variables obtained can be re-combined to yield the original sample covariance matrix.

\subsection{The case where $\mu$ is unknown}

We now turn to the case where the mean vector $\mu$ is unknown, and the data analyst is given access to $\bar Z_n$ and $S_n$ from \eqref{eq:suff-unknown-mean}, 
along with the sample size $n$.  The next result establishes that Algorithm~\ref{alg2}, a variant of Algorithm~\ref{alg1}, can be applied to thin $((n-1) S_n, \bar Z_n)$.  

\if\bmka1
\begin{algo}
\label{alg2} \textcolor{white}{.}\\
\indent \emph{Input:} a $p \times p$ positive semi-definite matrix $W$ of rank $r$, a $p$-vector $t$, and an integer $n > r$. 

\begin{enumerate}
    \item Perform an eigenvalue decomposition: $W=VD^2V^\top$, where $D$ is a $r \times r$ diagonal matrix, and $V$ is an orthogonal matrix of dimension $r \times p$.
    \item Draw $Q\sim\text{Unif}(O_{(n-1)\times r})$, where $O_{(n-1)\times r}=\{Q\in\mathbb{R}^{(n-1)\times r}: Q^\top Q=I_r\}$. 
    \item Return the $n\times p$ matrix $X=1_n t^\top + H\tilde X$, where $\tilde X=QDV^\top$ and $H\in\mathbb R^{n\times (n-1)}$ is an  orthogonal matrix such that $HH^\top=I_n-(1/n)1_n1_n^\top$. 
\end{enumerate}
\end{algo}
\else
\begin{mybox}[floatplacement=!h,label={alg2}]{ Decomposing a $p \times p$ positive semi-definite matrix $W$ of rank $r$, and a $p$-vector $t$,  into an $n \times p$ matrix $X$ for some $n>r$} 
\begin{enumerate}
    \item Perform an eigenvalue decomposition: $W=VD^2V^\top$, where $D$ is a $r \times r$ diagonal matrix, and $V$ is an orthogonal matrix of dimension $r \times p$.
    \item Draw $Q\sim\text{Unif}(O_{(n-1)\times r})$, where $O_{(n-1)\times r}=\{Q\in\mathbb{R}^{(n-1)\times r}: Q^\top Q=I_r\}$. 
    \item Return $X=1_n t^\top + H\tilde X$, where $\tilde X=QDV^\top$ and $H\in\mathbb R^{n\times (n-1)}$ is a non-random orthogonal matrix such that $HH^\top=I_n-(1/n)1_n1_n^\top$. 
\end{enumerate}
\end{mybox}
\fi

\begin{theorem}[Thinning the sample covariance and sample mean of independent Gaussians]
Suppose that we apply  Algorithm~\ref{alg2} to 
$( (n-1) S_n,\bar Z_n, n)$ to obtain an $n \times p$ matrix $X$, where $S_n$ and $\bar Z_n$ are defined in \eqref{eq:suff-unknown-mean} for  $Z\sim N_{n\times p}(1_n\mu^\top; I_n, \Sigma)$.  Then, (i)  
$X^\top (I_n - \frac{1}{n} 1_n 1_n^\top) X = (n-1)S_n$ and $\frac1{n} X^\top 1_n=\bar Z_n$, and (ii)  the rows
of $X$ are independent $N_p(\mu, \Sigma)$ random variables. Furthermore, let $C_1,\ldots,C_K$ denote a partition of the integers $\{1,\ldots,n\}$ such that $C_k \cap C_{k'} = \emptyset$ for any $k \neq k'$ and $\cup_{k=1}^K C_k = \{1,\ldots,n\}$, and define $\bar{X}^{(k)} := \frac{1}{|C_k|} \sum_{i \in C_k} X_i$ and 
 $S^{(k)} := \frac{1}{|C_k|-1} \sum_{i \in C_k} (X_i - \bar{X}^{(k)} ) (X_i - \bar{X}^{(k)}  ) ^\top $, 
  where $|C_k|$ is the number of elements in the set $C_k$. Then, (iii) $(|C_k|-1) S^{(k)} \sim \text{Wishart}_p(|C_k|-1, \Sigma)$, $\bar{X}^{(k)}  \sim N_p(\mu, \frac{1}{|C_k|} \Sigma)$,  and $\left(S^{(1)},\bar{X}^{(1)}\right),\ldots,\left(S^{(K)},\bar{X}^{(K)}\right)$   are independent.
\label{thm:unknownmean}
\end{theorem}

The proof of Theorem \ref{thm:unknownmean} is given in Supplement \ref{app:thm2}.

Theorem~\ref{thm:unknownmean} serves the same purpose as Corollary~\ref{cor:knownmean}, but operates in a context in which both $\mu$ and $\Sigma$ are unknown.  In this setting, one starts with a pair of sufficient statistics for the original unavailable sample, and produces $K$ independent pairs of these sufficient statistics.

We note that Algorithm~\ref{alg2} and Theorem~\ref{thm:unknownmean} are quite related to results in \cite{lindqvist2005monte}; however, their goals are not the same as ours.

\section{Numerical Results}
\subsection{Verification of Theorems~\ref{thm:root} and \ref{thm:unknownmean}}\label{sec:verification}

\begin{figure}
    \centering
    \includegraphics[width=1.0\linewidth]{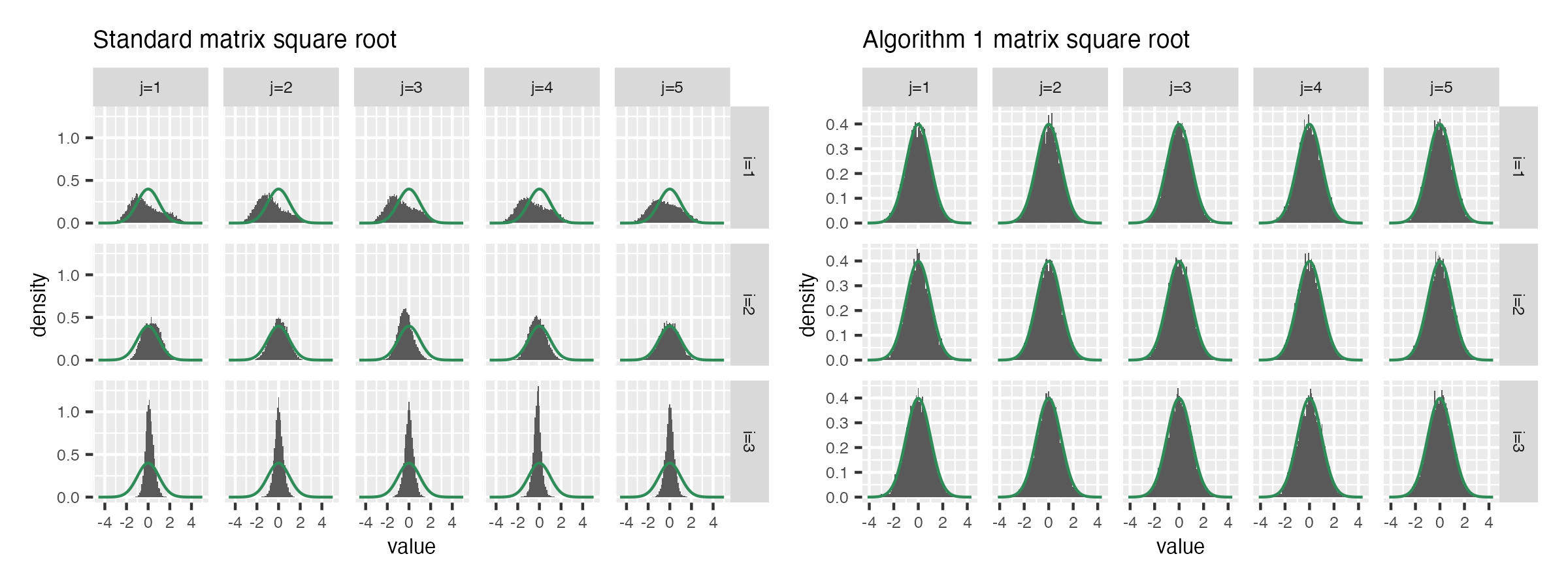}
    \caption{\em 
     For each of $10,000$ independent $\text{Wishart}_p(n, \Sigma)$ random matrices with $n=3$ and $p=5$, we generated two matrix square roots.  
    \emph{Left:} The elements of the matrix square root $DV^\top$ given by the eigendecomposition (see Step 1 of Algorithm~\ref{alg1}).  \emph{Right:} The elements of the matrix square root  given by  Step 3 of Algorithm~\ref{alg1}.  Only the latter yields a matrix square root for which the elements within the $j$th column follow a  $N(0, \Sigma_{jj})$ distribution (true distribution plotted in green), for $j = 1,\ldots,p$.}
    \label{fig:verification-alg1}
\end{figure}

Theorem~\ref{thm:root} establishes that applying Algorithm~\ref{alg1} to a Wishart matrix will generate a matrix square root whose rows are independent Gaussian random vectors.  In this section, we demonstrate in a  numerical example  that this is the case, and draw a contrast to another matrix square root that does not share this property. 

Setting $n=3$ and $p=5$, we first construct a $p \times p$ matrix $\Sigma$ with a Toeplitz structure, $\Sigma_{ij}=(1+|i-j|)^{-1}$, and draw $W\sim \text{Wishart}_p(n,\Sigma)$ by generating $Z\sim N_{n\times p}(\mathbf{0}_{n\times p},I_n,\Sigma)$ and then computing $W=Z^\top Z$. Let $V D^2 V^\top$ denote the eigendecomposition of $W$, and define $\breve{X}:= DV^\top$. Define $X$ to be the output of Step 3 of Algorithm~\ref{alg1} applied to $(W, n)$.  Figure~\ref{fig:verification-alg1} compares the entry-wise marginal distributions of $\breve X$ and $X$. 
 In particular, each panel contains an $n\times p$ array of histograms, the $(i,j)$th of which displays the distribution of $\breve{X}_{ij}$ (left) or $X_{ij}$ (right)  across 10,000 repetitions.   
  Superimposed on each histogram is the desired marginal distribution, $N(0,\Sigma_{jj})$.  We can see that the entries of $\breve{X}$ are far from normal, whereas the entries of $X$ have the correct marginals. 

In Appendix~\ref{app:numerical}, Figure~\ref{fig:verification-alg2} shows that when both $\Sigma$ and $\mu$ are unknown, each element $X_{ij}$ arising from Algorithm \ref{alg2} has the desired marginal distribution, $N(\mu_j,\Sigma_{jj})$.

\subsection{Application to post-selective inference in the graphical lasso}\label{sec:glasso}

The graphical lasso \citep{yuan2007model,banerjee2008model,rothman2008sparse,friedman2008sparse} estimator of the precision matrix $\Sigma^{-1}$ is 
\begin{equation}
\widehat\Omega_\lambda := \arg \min_{\Omega} \left\{ -\log \det \Omega + \text{trace} \left(\Omega  S_n  \right) + \lambda \| \Omega \|_1 \right\}, 
\label{eq:glasso}
\end{equation}
for $S_n$  defined in \eqref{eq:suff-unknown-mean}. 
Provided that $S_n$ arose from a sample of independent and identically distributed Gaussian random vectors, $\widehat\Omega_\lambda$ minimizes
 the negative log likelihood subject to an $\ell_1$ penalty on the elements of the precision matrix. Here, $\lambda$ is a nonnegative tuning parameter that determines the sparsity of $\widehat\Omega_\lambda$. In this section, we consider selecting $\lambda$ via cross-validation.

If we have access to the Gaussian random vectors $Z_1,\ldots,Z_n \sim N_p(\mu, \Sigma)$ used to compute $S_n$, then we can use sample splitting to instantiate a cross-validation scheme to select $\lambda$. In greater detail, let $C_1\ldots,C_K$ denote a partition of $\{1,\ldots,n\}$ such that $C_k \cap C_{k'} = \emptyset$ and $\cup_{k=1}^K C_k = \{1,\ldots,n\}$.
 Then, for $k=1,\ldots,K$, we define $S_{\text{SS}}^{(k)}=\frac{1}{|C_k|-1}\sum_{i\in C_k}(Z_i-\bar Z_{C_k})(Z_i-\bar Z_{C_k})^\top$ and $S_{\text{SS}}^{(-k)}=\frac{1}{n-|C_k|-1}\sum_{i\not\in C_k}(Z_i-\bar Z_{-C_k})(Z_i-\bar Z_{-C_k})^\top$ to be the sample covariance matrices computed on the observations in $C_k$ and on all but the observations in $C_k$, respectively (where $\bar Z_{C_k}$ and $\bar Z_{-C_k}$ are the corresponding sample means). We let $\widehat\Omega_{\lambda,\text{SS}}^{(-k)}$ denote the  graphical lasso estimator computed on $S_{\text{SS}}^{(-k)}$. 
 We select the value of $\lambda$ that minimizes 
$$
\ell_{\text{SS}}\left(\lambda\right)= \sum_{k=1}^K \left\{ -\log\det\widehat\Omega_{\lambda,\text{SS}}^{(-k)}+\trace\left(\widehat\Omega_{\lambda,\text{SS}}^{(-k)} S_{\text{SS}}^{(k)}\right) \right\}.
$$

Now, suppose that --- following the setup of this paper ---   we do not have access to $Z$ directly, but only to $S_n$ from \eqref{eq:suff-unknown-mean}. 
Consequently, cross-validation via sample splitting cannot be applied. 
 Instead, we apply  Algorithm~\ref{alg1} to $((n-1)S_n,n-1)$ to obtain an $(n-1) \times p$ matrix $X$; here, we use $n-1$ in place of $n$ because $S_n$ has rank $n-1$. By  Theorem~\ref{thm:root}, the rows of this  matrix are independent $N_p(0,\Sigma)$ random vectors. We then partition the indices $\{1,\ldots,n-1\}$  into $C_1,\ldots,C_K$, where $\cup_{k=1}^K C_k = \{1,\ldots,n-1\}$ and $C_k \cap C_{k'} = \emptyset$. We define  $S_{\text{DT}}^{(-k)} = \frac{1}{n-1-|C_k|} \sum_{i \notin C_k} X_i X_i^\top$
 and $S_{\text{DT}}^{(k)}=\frac{1}{|C_k|} \sum_{i \in C_k} X_i X_i^\top$. 
  Note that $|C_k| \cdot  S_{\text{DT}}^{(k)} \sim \text{Wishart}_p(|C_k|, \Sigma)$, and that   $S_{\text{DT}}^{(k)}$ and $S_{\text{DT}}^{(-k)}$ are independent.

 For $k=1,\ldots,K$, we let
 $\widehat\Omega_{\lambda,\text{DT}}^{(-k)}$ denote the graphical lasso estimator computed on $S_{\text{DT}}^{(-k)}$, with tuning parameter $\lambda$.
 We then select the value of $\lambda$ that minimizes 
$$
\ell_{\text{DT}}\left(\lambda\right)= \sum_{k=1}^K \left\{ -\log\det\widehat\Omega_{\lambda,\text{DT}}^{(-k)}+\trace\left(\widehat\Omega_{\lambda,\text{DT}}^{(-k)} S_{\text{DT}}^{(k)}\right) \right\}.
$$

\begin{figure}
    \centering
    \includegraphics[width=0.75\linewidth]{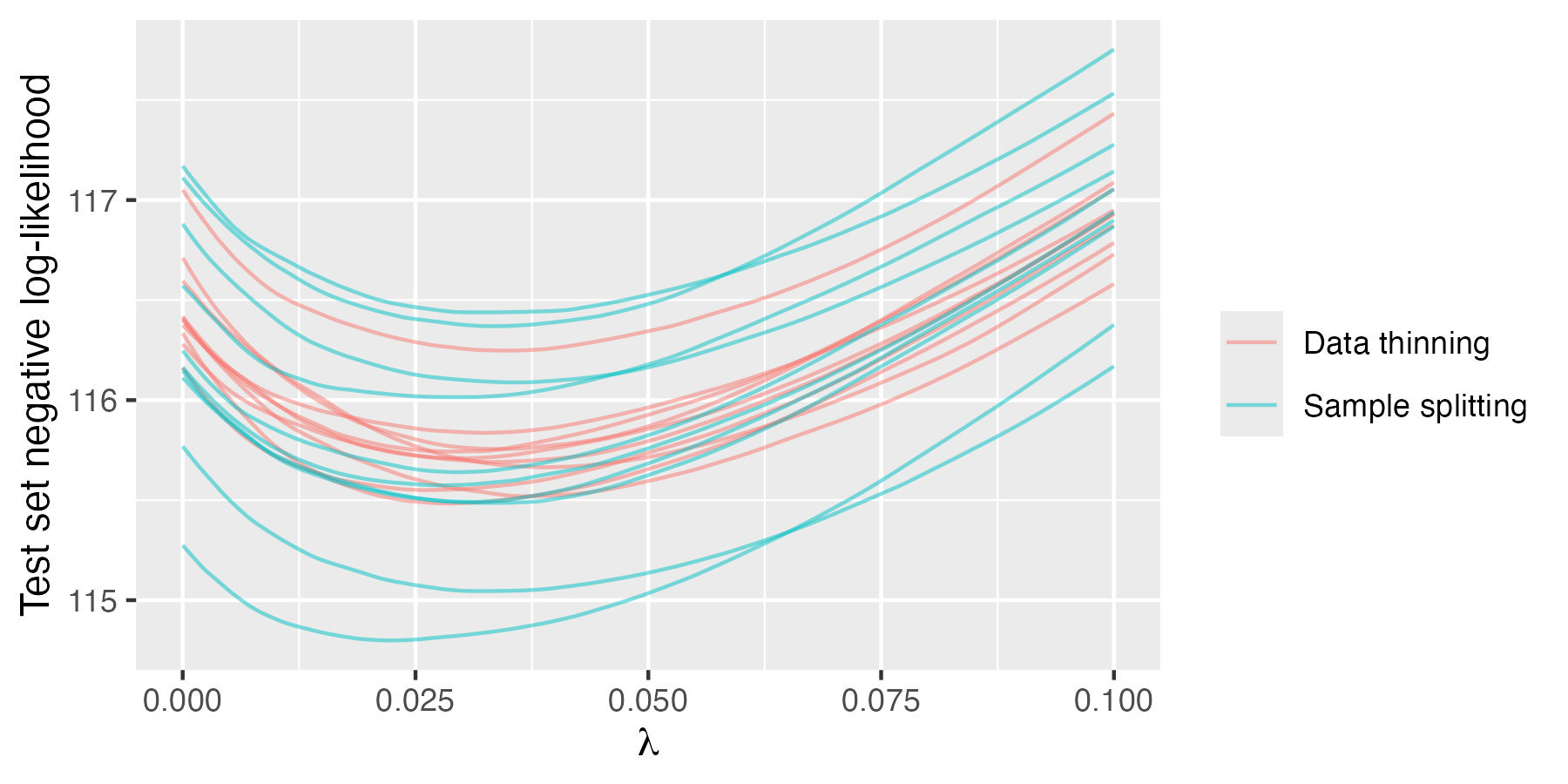}
    \caption{\em For the graphical lasso simulation described in Section~\ref{sec:glasso}, the figure displays ten realizations of  $\ell_{\text{SS}}(\lambda)$ and $\ell_{\text{DT}}(\lambda)$, the test set negative log-likelihoods for the sample splitting and data thinning approaches.  Computing $\ell_{\text{SS}}(\lambda)$ requires access to the Gaussian data matrix, whereas $\ell_{\text{DT}}(\lambda)$ requires access to only the sample covariance matrix. 
    }
    \label{fig:glasso}
\end{figure}

We now compare $\ell_{\text{SS}}(\lambda)$ and $\ell_{\text{DT}}(\lambda)$ in simulation.  We generate $n=250$ independent $N_p(\mu,\Sigma)$ random vectors where $p=10$, $\mu=0_{10}$, and $\Sigma^{-1}$ is block diagonal, with blocks $0.5\cdot I_4 + 0.5\cdot 1_41_4^\top$, $0.75\cdot I_4 + 0.25\cdot  1_41_4^\top$, and $I_2$. Figure~\ref{fig:glasso} displays $\ell_{\text{SS}}(\lambda)$ and $\ell_{\text{DT}}(\lambda)$ for $K=10$, for each of ten simulated datasets. We find that all curves are minimized when $\lambda \approx 0.025$. Therefore, data thinning selects the same tuning parameter as sample splitting, without requiring access to the individual-level data $Z_1,\ldots,Z_n$.

\section{Discussion} 

Arguments from \cite{neufeld2023data} and \cite{dharamshi2023generalized} suggest that it might be possible to 
thin a $\text{Wishart}_p(n; \Sigma)$ random matrix into $K$ independent Wishart random matrices $W^{(1)},\ldots,W^{(K)}$, with $W^{(k)} \mathop{\sim}\limits^{\mathrm{iid}} \text{Wishart}_p(n_k; \Sigma)$ for $n_1+ \ldots + n_k=n$, by sampling from the conditional distribution of $(W^{(1)},\ldots,W^{(K)})$ given $ \sum_{k=1}^K W^{(k)}$. Since $\sum_{k=1}^K W^{(k)}$ is sufficient for $\Sigma$, sampling from this conditional distribution does not require knowledge of $\Sigma$.  It turns out that this conditional distribution is closely related to the matrix variate Dirichlet distribution \citep{gupta2018matrix}. In fact, an alternative to the procedure described in Corollary \ref{cor:knownmean} can be obtained by sampling from this conditional distribution. 

Code to reproduce all numerical analyses in this note is available at \url{https://github.com/AmeerD/Wishart/}.

\bibliographystyle{agsm}
\bibliography{wishart}

\appendix

\section{Proof of Theorem \ref{thm:root}}
\label{sec:appendixA}

\begin{proof}

We start by noting that $X^\top X = (VDQ^\top )QDV^\top=VD^2V^\top= W$. It remains to show that $X=QDV^\top$ has the desired distribution.  
This will follow from some facts about the matrix normal.

Consider an $n\times p$ random matrix $Z\sim N_{n\times p}(\mathbf{0}_{n\times p},I_n,\Sigma)$, and denote its singular value decomposition as $Z=U(Z)D(Z)V(Z)^\top$.  By definition of the Wishart distribution, $Z^\top Z=V(Z)\left[ D(Z) \right] ^2V(Z)^\top$ has the same distribution as $W$.  Thus,
$D$ and $V$ from the eigenvalue decomposition of $W$ in Step 1 have the same joint  distribution as $D(Z)$ and $V(Z)$.  It remains to show  the following two claims:
 
\begin{list}{}{}
\item[Claim 1.] $U(Z)\perp (V(Z),D(Z))$; and 
\item[Claim 2.] $U(Z)$ is distributed uniformly on the $n\times r$ Stiefel manifold, where $r=\min(n,p)$.
\end{list}
Provided that these two claims hold, 
 $X=QDV^\top$ has the same distribution as $Z=U(Z) D(Z) V(Z)^\top$, and so  the proof is complete.

It remains to justify the two claims.
When $n \geq p$, both claims follow directly from \cite{james1954normal}. For $n<p$, we will show that the joint density of $\left(U(Z),D(Z),V(Z)\right)$ factors into the desired terms. Following the transformation $Z \to U(Z)D(Z)V(Z)^\top$, the joint density of $\left(U(Z),D(Z),V(Z)\right)$ simplifies as 
\begin{align*}
&f\left(U(Z),D(Z),V(Z)\right)\\
\propto &\exp(-\frac{1}{2}\trace[\Sigma^{-1}V(Z)D(Z)^2V(Z)^\top])|J|\\
=&\exp(-\frac{1}{2}\trace[\Sigma^{-1} V(Z)D(Z)^2V(Z)^\top])\prod_{i<j\le p}(d_i^2- d_j^2)|D(Z)|^{n-p}(d D(Z))(U(Z)^\top d U(Z))^\wedge(V(Z)^\top dV(Z))^\wedge
\end{align*} 
where $J$ indicates the Jacobian of the transformation, $d_i$ are the singular values, and $\wedge$ refers to the wedge product (see \citealt{svdjacobian} for details on the wedge product). For details on the derivation of the Jacobian, see \cite{srivastava2003singular} and \cite{svdjacobian}. Notice that $f\left(U(Z),D(Z),V(Z)\right)$ factors into $f\left(U(Z)\right)$ and $f\left(D(Z),V(Z)\right)$. This implies that $U(Z)$ is independent of $D(Z)$ and $V(Z)$. Further, the fact that $f\left(U(Z)\right)\propto(U(Z)^\top d U(Z))^\wedge$ implies that $U(Z)$ is uniformly distributed on the $n\times n$ Stiefel manifold \citep{anderson2003introduction, muirhead2009aspects}. Thus, both claims are proven when $n<p$.

\end{proof}

\section{Proof of Theorem \ref{thm:unknownmean}}
\label{app:thm2}

\begin{proof}
We begin by verifying (i): namely, that $X^\top (I_n - \frac{1}{n} 1_n 1_n^\top) X = (n-1)S_n$ and $\frac1{n} X^\top 1_n=\bar Z_n$.

First, recalling from Algorithm \ref{alg2} that $H\in\mathbb R^{n\times (n-1)}$ is an orthogonal matrix such that $HH^\top=I_n-(1/n)1_n1_n^\top$, note that $\|H^\top 1_n \|^2 = 1_n^\top H H^\top 1_n = 1_n^\top \left( I_n-(1/n)1_n1_n^\top \right) 1_n =0$. Therefore, $H^\top 1_n=0_{n-1}$.

Recalling the construction of $X$ from applying Algorithm \ref{alg2} with $(W,t)=((n-1)S_n,\bar Z_n)$, observe that $$\frac1{n} X^\top 1_n = \frac{1}{n} \left( \bar Z_n 1_n^\top +   VDQ^\top H^\top   \right) 1_n = \bar Z_n +   \frac{1}{n}    VDQ^\top  H^\top 1_n =\bar Z_n, $$
where the last equality follows from the fact that $H^\top 1_n=0_{n-1}$.
Furthermore,
$$
\left(I_n - \frac{1}{n} 1_n 1_n^\top \right) X = 
\left(I_n - \frac{1}{n} 1_n 1_n^\top \right) \left(1_n\bar Z_n^\top + HQDV^\top\right)
=   HQDV^\top
$$
since  $\left(I_n - \frac{1}{n} 1_n 1_n^\top \right) 1_n=0_{n}$ and $(I_n - \frac{1}{n} 1_n 1_n^\top)H=H$. Noting
 that $(I_n - \frac{1}{n} 1_n 1_n^\top)$ is idempotent, we have that $$X^\top \left(I_n - \frac{1}{n} 1_n 1_n^\top\right) X= VDQ^\top H^\top H Q DV^\top = VD^2 V^\top = (n-1)S_n,$$ where the 
second-to-last equality follows from the fact that $H^\top H = I_{n-1}$ and $Q^\top Q = I_{r}$, and the 
last equality follows from Step 1 of Algorithm~\ref{alg2}.

We will now establish (ii): namely, that $X$ has the desired distribution. First, note that $\bar{Z}_n \sim N_p(\mu, \Sigma/n)$. Next, observe that the $\tilde X$ generated in Step 3 of Algorithm \ref{alg2} is exactly the output of calling Algorithm~\ref{alg1} with $n-1$ in place of $n$ (this is allowed since Algorithm~\ref{alg2} requires $n>r$ whereas Algorithm~\ref{alg1} requires $n\ge r$).  Therefore, $\tilde X\sim N_{(n-1)\times p}(0;I_{n-1},\Sigma)$.  Recalling that $\bar{Z}_n \perp S_n$ and that $\tilde X$ depends only on $S_n$, we have that $\bar{Z}_n \perp \tilde X$.  Thus, $(\sqrt{n}\bar Z_n,\tilde X^\top)^\top\sim N_{n\times p}([\sqrt{n}\mu,0_{p\times (n-1)}]^\top;I_n,\Sigma)$. Writing $X=1_n\bar Z_n^\top + H\tilde X$ in matrix form, 
$$
X=\begin{pmatrix} \frac{1}{\sqrt{n}} 1_n  &  H \end{pmatrix} \begin{pmatrix}\sqrt{n}\bar Z_n^\top\\ \tilde X\end{pmatrix},
$$
establishes that $X$ is a linear transformation of a matrix normal and therefore is itself matrix normal, with mean $1_n\mu^\top$ and row and column covariance matrices $\begin{pmatrix} \frac{1}{\sqrt{n}} 1_n  &  H \end{pmatrix}\begin{pmatrix} \frac{1}{\sqrt{n}} 1_n  &  H \end{pmatrix}^\top=I_n$ and $\Sigma$, respectively.

It remains to establish (iii): namely, that $(|C_k|-1) S^{(k)} \sim \text{Wishart}_p(|C_k|-1, \Sigma)$ and $S^{(1)},\ldots,S^{(K)}$ are independent. The independence of $S^{(1)},\ldots,S^{(K)}$ follows immediately from the fact that the rows of $X$ are independent and $C_1,\ldots,C_K$ form a partition. To establish that $(|C_k|-1) S^{(k)} \sim \text{Wishart}_p(|C_k|-1, \Sigma)$, first observe that $(|C_k|-1) S^{(k)} = \sum_{i \in C_k} (X_i - \bar{X}^{(k)} ) (X_i - \bar{X}^{(k)}  ) ^\top = \left(X^{(k)}\right)^\top\left(I_{|C_k|}-\frac{1}{|C_k|}1_{|C_k|}1_{|C_k|}^\top\right)X^{(k)}$, where $X^{(k)}$ is the $|C_k|\times p$ submatrix of $X$ containing the rows of $X$ corresponding to $C_k$.  Furthermore, define $H^{(k)}$ to be a $|C_k| \times (|C_k| -1)$   orthogonal matrix with $H^{(k)}\left(H^{(k)}\right)^\top=I_{|C_k|}-\frac{1}{|C_k|}1_{|C_k|}1_{|C_k|}^\top$. Because $(H^{(k)})^\top X^{(k)} \sim N_{(|C_k|-1) \times p} ( 0; I_{|C_k|-1}; \Sigma)$, it follows that 
$\left(X^{(k)}\right)^\top\left(I_{|C_k|}-\frac{1}{|C_k|}1_{|C_k|}1_{|C_k|}^\top\right)X^{(k)}
= \left(H^{(k)}X^{(k)}\right)^\top\left(H^{(k)}X^{(k)}\right)$ is $\text{Wishart}_p(|C_k|-1, \Sigma)$.
\end{proof}

\section{Additional numerical experiments}
\label{app:numerical}

We conduct a second simulation study similar to that of Section \ref{sec:verification}, though with unknown $\mu$. Again, we set $n=3$ and $p=5$, and construct a length $p$ vector $\mu$ such that the $j$th entry $\mu_j=j$, and a $p \times p$ matrix $\Sigma$ with a Toeplitz structure, $\Sigma_{ij}=(1+|i-j|)^{-1}$. We then generate $Z\sim N_{n\times p}(\mathbf{0}_{n\times p},I_n,\Sigma)$, and compute $\bar Z_n = \frac{1}{n}Z^\top 1_n\sim N_p(\mu,\Sigma/n)$ and $W=Z^\top\left(I_n-\frac{1}{n}1_n1_n^\top\right) Z\sim \text{Wishart}_p(n-1,\Sigma)$. Let $X$ be the output of Step 3 of Algorithm~\ref{alg2} applied to $(W, \bar{Z}_n, n)$.  Figure~\ref{fig:verification-alg2} displays the marginal distributions of the elements of $X$. Each panel contains an $n\times p$ array of histograms, the $(i,j)$th of which displays the distribution of $X_{ij}$ across 10,000 repetitions. Superimposed on each histogram is the desired marginal distribution, $N(\mu_j,\Sigma_{jj})$.  We can see that the entries of $X$ have the correct marginals, thereby numerically verifying Theorem \ref{thm:unknownmean}.

\begin{figure}
    \centering
    \includegraphics[width=0.5\linewidth]{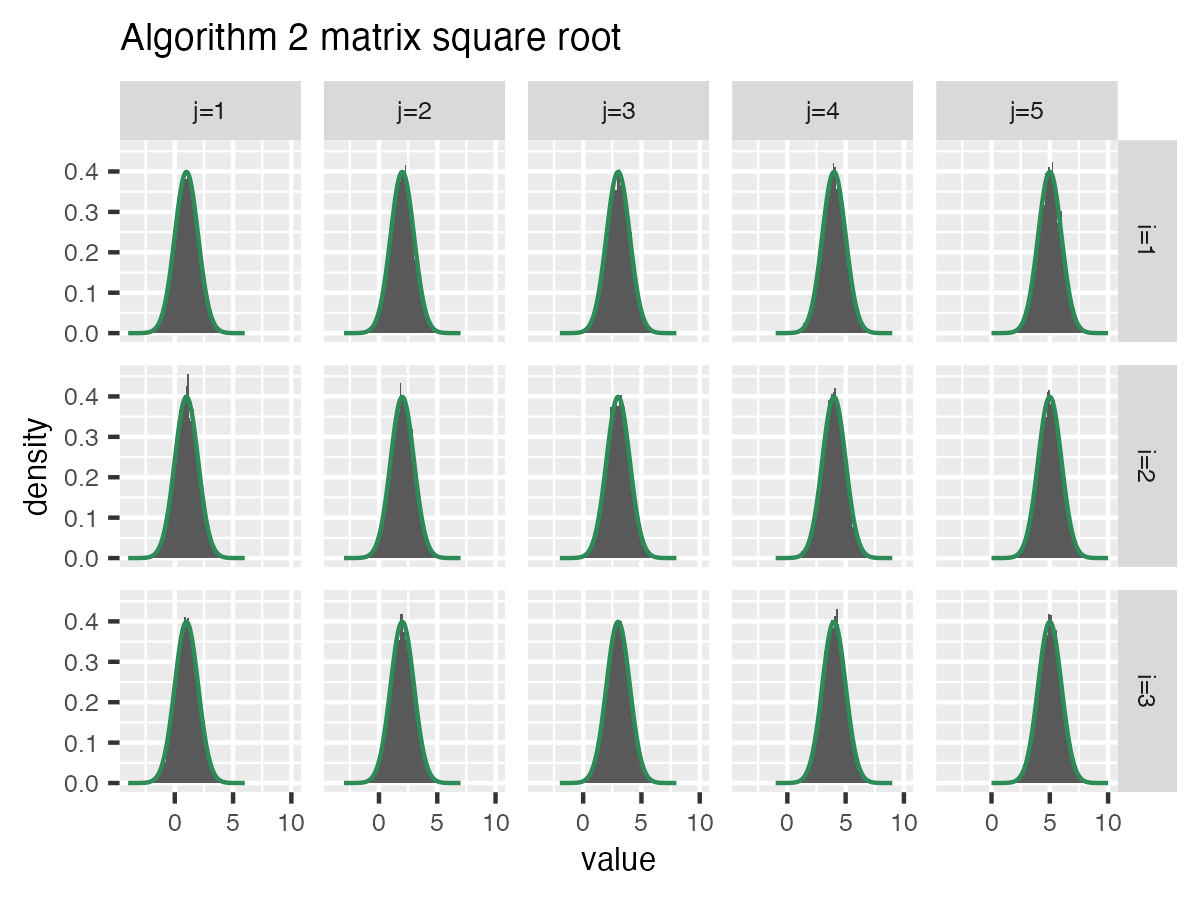}
    \caption{\em 
     For each of $10,000$ independent $N_p(\mu,\Sigma/n)$ random vector and $\text{Wishart}_p(n, \Sigma)$ random matrix pairs with $n=3$ and $p=5$, we generated the  matrix square root using Algorithm \ref{alg2}.  Each panel displays an element of the matrix square root  given by  Step 3 of Algorithm~\ref{alg2}.  The elements within the $j$th column follow a  $N(\mu_j, \Sigma_{jj})$ distribution (true distribution plotted in green), for $j = 1,\ldots,p$, in keeping with Theorem \ref{thm:unknownmean}.}
    \label{fig:verification-alg2}
\end{figure}

\end{document}